\documentclass[letterpaper, 10 pt, conference]{ieeeconf}
\usepackage{graphicx} 
\usepackage{cite}
\usepackage{xcolor}

\usepackage{amsthm}
\usepackage{amsmath}
\usepackage{float}
\usepackage{mathtools}
\mathtoolsset{showonlyrefs}
\usepackage{amssymb}
\usepackage{graphicx}
\usepackage{lmodern}
\usepackage{subcaption}
\usepackage{xcolor}
\usepackage{mathrsfs}
\usepackage{comment}
\usepackage{lastpage} 
\usepackage{extramarks} 
\usepackage[bb=dsserif]{mathalpha} 
\usepackage{epstopdf} 
\usepackage{amsmath}
\usepackage{algorithm}
\usepackage{algcompatible}
\usepackage{bbm}
\usepackage{enumerate}

\usepackage[utf8]{inputenc}
\usepackage[normalem]{ulem}

\usepackage{tikz}
\usepackage{pgfplots}
\usepackage{url}
\usepackage[]{subcaption}
\usepgfplotslibrary{groupplots,dateplot}
\usetikzlibrary{patterns,shapes.arrows, fadings}
\usetikzlibrary{automata} 
\usetikzlibrary{shapes.geometric}
\usetikzlibrary{calc,shadows}
\usetikzlibrary{positioning} 
\usetikzlibrary{arrows, arrows.meta} 
\tikzset{node distance=2.5cm, 
every state/.style={ 
semithick,
fill=gray!10},
initial text={}, 
double distance=2pt, 
every edge/.style={ 
draw,
->,>=stealth', 
auto,
semithick}}
\usetikzlibrary{spy}
\pgfplotsset{compat=newest}
\newlength\figH
\newlength\figW
\theoremstyle{definition}
\newtheorem{definition}{\protect\definitionname}

\newtheorem{lemma}{\protect\lemmaname}

\newtheorem{mechanism}{\protect\mechanismname}

\theoremstyle{plain}
\newtheorem{theorem}{\protect\theoremname}
\newtheorem{problem}{Problem}

\DeclareMathOperator{\adj}{Adj}

\DeclareMathOperator{\Binom}{Binom}

\newcommand{\pushright}[1]{\ifmeasuring@#1\else\omit\hfill$\displaystyle#1$\fi\ignorespaces}

\makeatother

\providecommand{\assumptionname}{Assumption}
\providecommand{\definitionname}{Definition}
\providecommand{\lemmaname}{Lemma}
\providecommand{\theoremname}{Theorem}
\providecommand{\mechanismname}{Mechanism}
\providecommand{\propositionname}{Proposition}

\newcommand{\cov}[2]{\text{Cov}\left(#1, #2\right)}
\newcommand{\E}[1]{\mathbb{E}\left[#1\right]}

\newcommand{\Eem}[1]{\mathbb{E}_{Exp}\left[#1\right]}

\newcommand{\Adjacent}[2]{\adj_{n,b}(#1, #2)=1}

\newcommand{\ones}{\mathbb{1}}

\newcommand{\prob}[1]{\mathbb{P}\left(#1 \right)}
\newcommand{\probem}[1]{\mathbb{P}_{Exp}\left(#1 \right)}

\newboolean{show_comments}
\setboolean{show_comments}{true} 
\ifthenelse{\boolean{show_comments}}{
    \newcommand{\Alex}[1]{\textcolor{blue}{Alex:~#1}} 
    \newcommand{\mh}[1]{\textcolor{red}{MH:~#1}} 
    \newcommand{\brandon}[1]{\textcolor{orange}{bf:~#1}} 
} {
    \newcommand{\Alex}[1]{}
    \newcommand{\mh}[1]{}
    \newcommand{\brandon}[1]{}
}

\IEEEoverridecommandlockouts
\title{ \LARGE \bf Differential Privacy for Symbolic Trajectories \\via the Permute-and-Flip Mechanism}
\author{Alexander Benvenuti, Huaiyuan Rao, Matthew Hale
\thanks{
School of  Electrical and Computer Engineering, Georgia Institute of Technology, Atlanta, GA USA.
Emails: \texttt{\{abenvenuti3, hrao43, matthale\}@gatech.edu}.
}
\thanks{
This work was partially supported by AFRL under grant FA8651-23-F-A008, 
NSF under CAREER grant 2422260 and
Graduate Research Fellowship grant DGE-2039655, 
ONR under grant N00014-21-1-2502, and AFOSR under grant FA9550-19-1-0169. Any opinions, findings and conclusions or recommendations expressed in this material are those of the authors and do not necessarily reflect the views of sponsoring agencies.
}
}
\begin{document}
\maketitle
\begin{abstract}
  Privacy techniques have been developed for data-driven systems,   
  but systems with non-numeric data cannot use typical noise-adding techniques. 
  Therefore, we develop a new mechanism for privatizing state trajectories of symbolic systems that may be 
  represented as words over a finite alphabet. 
  Such systems include Markov chains, Markov decision processes, and finite-state automata, and
  we protect their symbolic trajectories with differential privacy. 
  The mechanism we develop 
  randomly selects a private approximation to be released in place of the original sensitive word, 
  with a bias towards low-error private words.
 This work is based on the permute-and-flip mechanism for differential privacy, which can be applied
  to non-numeric data. However, a na\"{\i}ve implementation 
  would have to enumerate an exponentially large list of words to generate a private word.
  As a result, we develop a new mechanism that generates private words without ever needing to enumerate such a list.
  We prove that the accuracy of our mechanism is never worse than the prior state of the art, and 
  we empirically show on a real traffic dataset that it 
  introduces up to~$55$\% less error than the prior state of the art under a conventional 
  privacy implementation. 
\end{abstract}
\section{Introduction}
With the proliferation of data-driven systems, interest has arisen in 
developing 
techniques to privatize the user data they require, e.g., in traffic systems~\cite{glancy2012privacy, hassan2019differential} and smart power grids~\cite{guan2018privacy}. Absent such protections, observers may make accurate inferences about sensitive information, such as home occupancy~\cite{lundstrom2016detecting} and daily traveling routines~\cite{gong2011iterative}. These systems require user data to function, 
and thus there exists a need to preserve users' privacy 
across a broad range of systems. 

Therefore, in this work we develop a framework for privatizing trajectories generated by symbolic systems.
Symbolic systems generate sequences of non-numeric data, which are often represented as words or strings over a finite alphabet. 
An example class of such systems is Markov chains, in which trajectories are sequences of states that may represent user locations in domestic time use~\cite{widen2009combined}, the intersections 
traversed by a user 
in a traffic system~\cite{besenczi2020large}, or websites a user has visited~\cite{rendle2010factorizing}. Such trajectories may be sensitive, and 
we develop a privacy framework for them. 

We use differential privacy to develop this framework. Differential privacy is a statistical notion of privacy that was developed in the computer science literature to protect sensitive database entries when databases are queried~\cite{dwork2014algorithmic}. Differential privacy has been used in data-driven systems because of its desirable properties, specifically that it is (i) robust to side information, in that learning additional information about the underlying sensitive data does not substantially weak differential privacy, and (ii) immune to post-processing, in that post-hoc computations on differentially private data do not weaken its protections. These properties have led to the development of differential privacy frameworks for filtering~\cite{le2013differentially}, multi-agent control~\cite{yazdani2022differentially, hawkins2020differentially, chen2023differential, benvenuti2023differentially}, and optimization~\cite{huang2015differentially, han2016differentially, hale2017cloud, benvenuti2024guaranteed, benvenuti2025differentially}, among others. 
These works and others in control and optimization have 
implemented differential privacy by using the Gaussian and/or Laplace mechanisms to add noise to sensitive
data (or functions thereof). 

Symbolic systems present a challenge because noise cannot be added to non-numeric data. 
Prior work in~\cite{jones2019towards,chen2023differentialsymbolic} developed privacy
for symbolic systems by implementing the exponential mechanism, which is 
designed for non-numeric data. 
More recently, the permute-and-flip mechanism~\cite{mckenna2020permute} has emerged as an improvement on the exponential mechanism. It offers accuracy that is equal to or better than the exponential mechanism in all cases,
and it is often the optimal mechanism for generating private outputs (in a sense we make precise
in Section~\ref{sec:mech_design}). 

However, the permute-and-flip mechanism can have prohibitive time complexity over large output spaces, and 
the development of efficient, domain-specific implementations is an open problem~\cite{mckenna2020permute}. 
In this work, we solve that open problem 
for symbolic systems. 
Our approach consists of first randomly selecting a number of errors for the private output word to have, and 
then constructing an automaton that uniformly samples from the set of all private output words 
that are (i) the same length as the sensitive input word and (ii) have the selected number errors. 


To summarize, we make the following contributions:
\begin{itemize}
    \item We develop a mechanism for privatizing symbolic trajectories based on the permute-and-flip mechanism
    (Mechanism~\ref{mech:prob1}, Theorem~\ref{thm:mech1_priv}).
    \item We bound the expected Hamming distance between a sensitive trajectory and its privatized form, 
    and we show it is never worse than the previous state of the art (Theorem~\ref{thm:accuracy_bounds}).
    \item We specialize this mechanism to Markov chains (Mechanism~\ref{mech:prob3}, Theorem~\ref{thm:mech2}). 
    \item We empirically evaluate our approach on a real traffic dataset and show it incurs~$55\%$ 
    less error than the prior state of the art (Section~\ref{sec:sims}).
\end{itemize}

\subsection{Related Work}
To privatize symbolic trajectories,~\cite{jones2019towards,chen2023differentialsymbolic} develop efficient mechanisms based on the the exponential mechanism~\cite{dwork2014algorithmic}. 
The current paper draws in part from~\cite{chen2023differentialsymbolic} by using its techniques 
for efficiently computing samples of a probability distribution over a large set of words. 
However, the work in this paper fundamentally differs from~\cite{chen2023differentialsymbolic} by developing
and sampling from an entirely different distribution when implementing differential privacy. 
We show that this approach results in lower expected error. 

Several forms of privacy for Markov chains have previously been considered in~\cite{fallin2023differential, guner2023learning, benvenuti2026differentiallyprivatedatadrivenmarkov}. However, these works all consider privacy for the transition probabilities and data used to compute transition probabilities, while we consider privacy for trajectories.

\subsection{Notation}
We use~$\mathbb{N}$ to denote the set of non-negative integers and~$\mathbb{N}^+$ to denote the set of 
positive integers. For~$N\in\mathbb{N}^+$, we define~$[N] = \{1,\ldots,N\}$. We use~$\sum_{S\subseteq \mathcal{R}}$ to denote the sum over all subsets~$S$ of some finite set~$\mathcal{R}$. We use~$|B|$ to denote the cardinality of a finite set~$B$. We use~$\ones^{n}$ to denote a row vector of all~$1$'s of length~$n$. 
\section{Background and Problem Statements}
\subsection{Symbolic Systems}
Symbolic systems may be defined in terms of finite-state automata, which we define next.

\begin{definition}[Finite State Automaton]
    A finite state automaton (FSA) is a tuple~$A = (Q, \Sigma, q^0, \delta, F)$, where~$Q$ is a set of states,~$\Sigma$ is an input alphabet,~$q^0\in Q$ is the initial state,~$\delta: Q\times \Sigma \to Q$ is transition function between states, and~$F\subseteq Q$ is the set of accepting states. We use~$Q^n$ to denote the set of all state sequences of length~$n$ and~$Q^*$ to denote the set of all finite state sequences. We similarly define~$\Sigma^n$ as the set of all words of length~$n$ over~$\Sigma$.
\end{definition}
Given an FSA~$A = (Q, \Sigma, q^0, \delta, F)$, if the transition function~$\delta$ is nondeterministic, i.e.,~$\delta: Q\times \Sigma \to 2^Q$, then~$A$ is called a nondeterministic finite-state automaton (NFA). A run on an NFA~$A = (Q, \Sigma, q^0, \delta, F)$ induced by a word~$w = \sigma_1\cdots\sigma_n\in\Sigma^n$ is a sequence of states~$q = q_0\cdots q_n\in Q^n$ such that~$q_0 = q^0$ and~$q_{i+1} \in \delta(q_i, \sigma_{i+1})$. 
The automaton~$A$ accepts a word~$w$ if the final state of the induced run is an accepting state~$q_n\in F$. The set of all words accepted by the automaton~$A$ is its language, denoted by~$\mathcal{L}(A)$.

Throughout this work, we use the Hamming distance to compare two words of the same length, 
and it is denoted~$d(w, v)$ for words~$w, v$, which is equal to the number of positions in which the corresponding symbols differ, i.e.,~$d(w, v) = |\{i\mid w_i\neq v_i\}|$.

\subsection{Markov Chains}
Markov chains are a widely used class of symbolic systems. 
    A discrete time stochastic process~$(Y_t)_{t\in\mathbb{N}}$ on a state space~$\mathcal{Y}$ is a Markov chain if it satisfies the Markov property, i.e.,
    $\prob{Y_{t+1} = y_{t+1} \mid Y_1 = y_1, Y_2 = y_2, \ldots, Y_t = y_t}  = \prob{Y_{t+1} = y_{t+1} \mid Y_t = y_t}.$
Throughout this work, we denote Markov chains by a tuple~$(\mathcal{Y}, T, y_0)$, where~$\mathcal{Y}$ is the state space,~$T$ is the transition probability matrix, and~$y_0\in \mathcal{Y}$ is the initial state. The probability of transitioning from state~$y_i$ to state~$y_j$ is~$T_{y_i,y_j} = \prob{y_j\mid y_i}$. 
State~$y_j$ is \emph{feasible} from state~$y_i$ if~$T_{y_i,y_j}>0$. 
For~$n \in \mathbb{N}^+$, fix a time horizon~$[n]$, and let~$\mathcal{Y}^n$ denote the set of all sequences of length~$n$ with initial state~$y_0$.  Any such sequence can be identified with a word~$w = y_0\cdots y_n\in\mathcal{Y}^n$. 
The word~$w$ is feasible 
for a given Markov chain 
if $T_{y_{t+1},y_t}>0$ for all~$t\in[n-1]$. 
The set of all feasible words of length~$n$ is 
denoted $\mathcal{L}(\mathcal{Y}^n)$. 

\subsection{Differential Privacy}
Given an alphabet~$\Sigma$, we provide privacy to a sensitive word~$w\in\Sigma^n$ by using differential privacy. 
The goal of differential privacy is to make ``similar'' pieces of data appear approximately indistinguishable. The notion of ``similar'' is defined by an adjacency relation.
\begin{definition}[Word Adjacency;~\cite{jones2019towards}]\label{def:word_adj}
    Fix a length~$n\in\mathbb{N}^+$ and an adjacency parameter~$b\in\mathbb{N}^+$.     
    Two words~$w, v\in\Sigma^n$ are said to be adjacent if~$d(w, v)\leq b$. 
\end{definition}

Differential privacy is enforced by a randomized mapping called a ``mechanism'', which we denote by~$\mathcal{M}$. 

\begin{definition}[Word Differential Privacy; \cite{jones2019towards}]\label{def:word_dp}
    Fix a probability space~$(\Omega, \mathcal{F}, \mathbb{P})$, an adjacency parameter~$b\in\mathbb{N}^+$, a word length~$n\in\mathbb{N}^+$, and a privacy parameter~$\epsilon>0$. A mechanism~$\mathcal{M}: \Sigma^n\times \Omega\to \Sigma^n$ is word~$\epsilon$-differentially private if, for all~$w, v$ adjacent in the sense of Definition~\ref{def:word_adj} and all~$L\subseteq \Sigma^n$,
     $   \prob{\mathcal{M}(w)\in L}\leq e^{\epsilon}\prob{\mathcal{M}(v)\in L}.$
\end{definition}

The parameter~$\epsilon$ quantifies the strength of privacy, and a smaller value of~$\epsilon$ implies stronger privacy. Typical values of~$\epsilon$ range from 0.1 to  10~\cite{hsu2014differential}.

\subsection{Problem Statements}

\begin{problem}\label{prob:mech_design}
    Design a privacy mechanism to generate private approximations for symbolic trajectories.
\end{problem}

\begin{problem}\label{prob:accuracy}
    Bound the accuracy of the proposed mechanism in terms of the privacy parameter~$\epsilon$, the length of the word~$n$, and the size of the alphabet~$|\Sigma|$. 
\end{problem}

\begin{problem}\label{prob:markov_chains}
    Extend the mechanism from Problem~\ref{prob:mech_design} to privatize trajectories of a Markov chain while ensuring that all private output words are feasible for the given Markov 
    chain.    
\end{problem}

\begin{problem}\label{prob:sims}
    Empirically compare the mechanism from Problem~\ref{prob:markov_chains} to the state
of the art and quantify its improvement in accuracy.   
\end{problem}

\section{Mechanism Design and Analysis}\label{sec:mech_design}
In this section, we solve Problems~\ref{prob:mech_design} and~\ref{prob:accuracy}. 
Given an NFA~$A = (Q, \Sigma, q^0, \delta, F)$ and a sensitive word~$w = \sigma_1\cdots\sigma_n\in\Sigma^n$, to enforce word differential privacy we randomly generate a private 
word~$w' = \sigma'_1\cdots \sigma'_n\in \Sigma^n$. To do so, we input the sensitive word~$w$ into an NFA we define in this section.
We design that NFA so that its state trajectories are
private approximations to~$w$, and one of those state trajectories is 
used as the private output word~$w'$ generated by the mechanism we develop. 

\subsection{Mechanism Design}
This subsection solves Problem~\ref{prob:mech_design}.
As described in the Introduction, we use the permute-and-flip mechanism to privatize
non-numerical data. It assigns probabilities to all possible outputs based on a utility score, which encodes how well a private output word~$w'$ approximates a sensitive input word~$w$. Throughout this work, we use the utility function 
\begin{equation} \label{eq:udef}
u(w, w') = -d(w, w'), 
\end{equation}
which 
encodes the fact that~$w'$ is a better approximation for a sensitive input word~$w$
if it differs from~$w$ by fewer symbols. 

In the next lemma and definition, we consider private output words that are in some pre-specified set $L \subseteq \Sigma^n$, which is the set of possible
private output words that are feasible for a given system. 
If all words in~$\Sigma^n$ are feasible,
then we may set $L = \Sigma^n$. However, some state transitions are not possible in some systems, such as Markov chains, and we allow for~$L \neq \Sigma^n$ for problems in which only a 
subset of the words in~$\Sigma^n$ 
are feasible.
To calibrate the privacy mechanism we use, 
we first require the sensitivity of the utility function. 

\begin{lemma}[Sensitivity; \cite{jones2019towards}]\label{lem:sensitivity}
    Fix an alphabet~$\Sigma$, a word length~$n\in\mathbb{N}^+$, a set~$L\subseteq \Sigma^n$, and an adjacency parameter~$b\in\mathbb{N}^+$. Then the sensitivity of the utility function~$u$ from~\eqref{eq:udef} is
    \begin{equation}
        \Delta u = \max_{v\in L}\max_{\substack{w_1, w_2\in L\\ \Adjacent{w_1}{w_2}}} |u(w_1, v) - u(w_2, v)| \leq b.
    \end{equation}
\end{lemma}



Next we formally state the permute-and-flip mechanism
\begin{definition}[Permute-and-Flip; \cite{mckenna2020permute}]\label{def:PF}
    Fix an alphabet~$\Sigma$, a word length~$n\in\mathbb{N}^+$, a set~$L\subseteq \Sigma^n$, and an adjacency parameter~$b\in\mathbb{N}^+$. For a sensitive input word~$w\in L\subseteq \Sigma^n$, the permute-and-flip mechanism~$\mathcal{M}_{PF}$ selects the private output word~$w'\in L$ with probability
    \begin{equation}\label{eq:pf_pmf}
        \prob{\mathcal{M}_{PF}(w) = w'} = \exp\left(\frac{\epsilon u(w, w')}{2\Delta u}\right)\Psi(b, L, w'),
    \end{equation}
    where
\begin{equation}\label{eq:psi}
    \Psi(b, L, w') = \sum_{k=0}^{|L|}\frac{(-1)^k}{k+1}\sum_{\substack{S\subseteq L\\|S| = k\\ w'\not\in S}} \prod_{s\in S} \exp\left(\frac{\epsilon u(w, s)}{2\Delta u}\right).
\end{equation}
\end{definition}

A direct implementation of the permute-and-flip mechanism requires knowledge of the Hamming distance between the sensitive input word and every word in~$\Sigma^n$ to compute~$\Psi(b, L, w')$. There are~$m^n$ total strings of length~$n$ on an alphabet of~$m$ symbols, and enumerating all such strings 
is prohibitive for large sensitive input words or large alphabets. 
The authors in~\cite{chen2023differentialsymbolic} faced the same challenge when implementing the exponential mechanism for privatizing words, and they developed an efficient implementation for selecting a private output word which does not require enumerating all strings of length~$n$. 

We seek to develop a similar framework that implements the permute-and-flip mechanism. Inspired by the approach in~\cite{chen2023differentialsymbolic}, we use the modified Hamming distance automaton.

\begin{definition}[Modified Hamming distance automaton; \cite{chen2023differentialsymbolic}]\label{def:MNFA}
    Fix an alphabet~$\Sigma$ and a word length~$n\in\mathbb{N}^+$. For a word~$x\in\Sigma^n$ and a distance~$\ell\in\mathbb{N}$, the modified Hamming distance NFA (MNFA) is an NFA~$A_{x, \ell} = (Q_{x, \ell}, \Sigma, q^0, \delta, F_{n, \ell})$ such that~$\mathcal{L}(A_{x, \ell})$ is the set of all words of length~$n$ with Hamming distance from~$x$ equal to~$\ell$. Each state~$q\in Q_{x, \ell}$ can transfer to another state by a policy~$\mu(\cdot, \cdot\mid q_i) : Q_{x, \ell} \times \Sigma\to[0, 1]$, where~$\mu(q_{i+1}, \sigma_{i+1}\mid q_i)$ is the probability that the input symbol~$\sigma_{i+1}$ causes a transition from state~$q_i$ to state~$q_{i+1}$. 
\end{definition}

Below, Mechanism~\ref{mech:prob1} defines a privacy mechanism that uses an MNFA to generate private output words. 
Although the MNFA model 
in Definition~\ref{def:MNFA}
does not define an output, its state trajectories are the outputs
of the mechanism we define. With an abuse of terminology, we sometimes
call state trajectories of the MNFA ``outputs'' when discussing
privacy. 

An MNFA begins with an empty output word and appends symbols one at a time 
until it produces a word in~$\mathcal{L}(A_{x,\ell})$. Such a word has length~$n$
and has~$\ell$ errors, which are differences between the output word and the input word it approximates. 
A state~$q_{i, j}$ contains the current length of the private output word that is being assembled, which is~$i$, and the number of errors currently in it, which is~$j$. Then, based on the policy~$\mu$, a new symbol is appended that 
increments~$i$ and possibly~$j$. 
This process is repeated until a private output word is generated with length~$n$ and~$\ell$ errors.

\begin{algorithm}[t]
    \caption{Constructing MNFA and Policy Synthesis~\cite{chen2023differentialsymbolic}}
    \label{algo:MNFA}
    \begin{algorithmic}[1]
    \STATEx \textbf{Inputs}: Sensitive word $w$ with length~$n$, transition function $\delta$, accepting set $\{q_{n, \ell}\}$
    \STATEx \textbf{Outputs}: Policy~$\mu$
    \vspace{1mm}
    \STATE $V(q_{n, \ell}) = 1$
    \STATE $CurrQ = \{q_{n, \ell}\}$
    \STATE $ActiveQ = \{\}$
    \FOR{$iteration\in [n]$}
    \FOR{$q'\in CurrQ$}
    \FOR{$(q, \sigma)~s.t.~q'\in \delta(q, \sigma)$}
         \STATE $V(q) = \sum_{q''\mid \exists \alpha, q''\in\delta(q, \alpha)} V(q'')$
         \STATE $\mu(q', \sigma \mid q) = \frac{V(q')}{V(q)}$
         \STATE $ActiveQ = ActiveQ\cup\{q\}$ 
         \ENDFOR
    \ENDFOR
    \STATE $CurrQ = ActiveQ$
    \STATE $ActiveQ = \{\}$
    \ENDFOR
    \end{algorithmic}
\end{algorithm}
\begin{figure}
    \centering
    \resizebox{\linewidth}{!}{
    \begin{tikzpicture}[
        ->, 
        >={Stealth[length=2.5mm]},
        node distance=2.5cm and 2.5cm, 
        on grid, 
        auto, 
        every state/.style={
            draw, 
            circle, 
            minimum size=1.2cm, 
            inner sep=0pt
        },
        initial text=start 
    ]

    
    \node[state, initial] (00) at (0,0) {$0^0$};

    \node[state] (10) at (3,0) {$1^0$};
    \node[state] (11) at (3,2.5) {$1^1$};

    \node[state] (21) at (6,2.5) {$2^1$};
    \node[state] (22) at (6,5) {$2^2$};

    \node[state, accepting] (32) at (9,5) {$3^2$};

    
    \path
        (00) edge node {$\textcolor{red}{\frac{1}{3}} a$} (10)
        (00) edge node {$\textcolor{red}{\frac{1}{3}}b,\textcolor{red}{\frac{1}{3}}c$} (11)

        (11) edge node {$\textcolor{red}{\frac{1}{2}}b$} (21)
        (11) edge node {$\textcolor{red}{\frac{1}{4}}a,\textcolor{red}{\frac{1}{4}}c$} (22)

        (10) edge node {$\textcolor{red}{\frac{1}{2}}a,\textcolor{red}{\frac{1}{2}}c$} (21)

        (22) edge node {$\textcolor{red}{1}c$} (32)

        (21) edge node[swap] {$\textcolor{red}{\frac{1}{2}}a,\textcolor{red}{\frac{1}{2}}b$} (32) 

    ;

    \end{tikzpicture}
}
    \caption{Modified Hamming distance automaton for all private output words of length~$3$ and Hamming distance~$2$ from the sensitive input word~$abc$ over the alphabet~$\{a, b, c\}$. Each circle represents a state where the base of the number in the circle is the current length of the private output word and the superscript is the number of errors made
    on the way to reaching that state. Each arrow represents a feasible transition, and the state with the double circle represents the accepting state. Probabilities from the policy~$\mu$ are in red.}
    \label{fig:MNFA}
\end{figure}
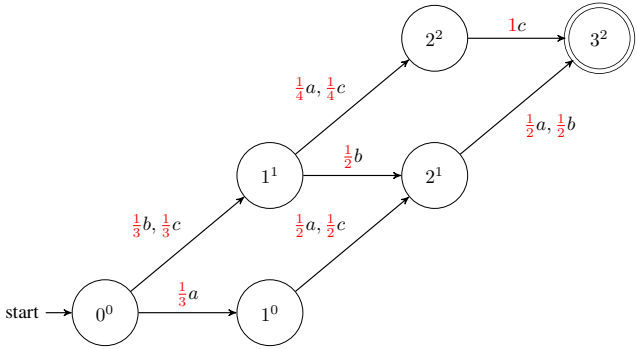

 Algorithm~\ref{algo:MNFA} is used to construct the MNFA 
  and synthesize the policy~$\mu$. 
 Algorithm~\ref{algo:MNFA} works by assigning a function~$V:Q_{w, \ell}\to \mathbb{N}$ such that~$V(q)$ is the number of unique paths in the MNFA from the state~$q\in Q_{w, \ell}$ that end in the accepting state~$q_{n, \ell}$. The probability of appending the symbol~$\sigma\in\Sigma$ at a state~$q$ is~$V(\delta(q, \sigma))/V(q)$, 
 which is the fraction of the paths in the MNFA from~$q$ to~$q_{n, \ell}$ that pass through some
  ~$q'\in\delta(q, \sigma)$.
 This procedure uniformly samples a private output word of length~$n$ with~$\ell$ errors from the set of all words of length~$n$ with~$\ell$ errors. 
 Figure~\ref{fig:MNFA} provides an illustration of Definition~\ref{def:MNFA} and Algorithm~\ref{algo:MNFA}. To compute the transition function~$\delta$ for the MNFA, we modify the Levenshtein automaton construction in~\cite{schulz2002fast} to use the Hamming distance.

 The implementation of the permute-and-flip mechanism is as follows. 
 First, we randomly select a Hamming distance between the sensitive input word~$w$ and the private output word~$w'$, which 
 is~$\ell$. Then, we construct a MNFA and compute a policy which, when executed on the MNFA, uniformly samples from the set of all words with Hamming distance~$\ell$ from the sensitive input word. To formally state the mechanism, we use
 \begin{equation}
 M(n, m, \ell) := \binom{n}{\ell}(m-1)^{\ell} 
 \end{equation}
 and
\begin{equation}
     \mathscr{L} := \left[0, \ones^{M(n, m, 1)}, \cdots, n \ones^{M(n, m, n)}\right]^T\in \mathbb{R}^{|\Sigma^n|}.
\end{equation}
Here,~$M(n, m, \ell)$ is the number of words of length~$n$ on an alphabet of size~$m$ that 
are
Hamming distance~$\ell$ from the sensitive input word~$w$.
The vector~$\mathscr{L}$ contains the Hamming distance to the sensitive input word~$w$ from every 
candidate output word~$w'\in\Sigma^n$. 
However, the construction of~$\mathscr{L}$ 
requires only elementary combinatorial terms and
does not require the explicit enumeration of output words to compute these Hamming distances. 
Additionally, let~$i{(\ell)}$ denote the smallest index~$j$ of~$\mathscr{L}$ such that~$\mathscr{L}_j = \ell$.

\begin{mechanism}[Solution to Problem 1]\label{mech:prob1}
Fix a probability space~$(\Omega, \mathcal{F}, \mathbb{P})$ and an adjacency parameter~$b\in\mathbb{N}^+$. Let an alphabet~$\Sigma$ and sensitive input word~$w = \sigma_1\cdots\sigma_n\in\Sigma^n$ be given, let~$N = |\Sigma^n|$, and let~$I = [N]$. The mechanism~$\mathcal{M}_{1}: \Sigma^n\times \Omega \to \Sigma^n$ selects a private output word by:
(i) drawing a Hamming distance~$\ell$ from the distribution 
    \begin{equation}\label{eq:mech1_pmf}
    \prob{\ell; w, b} = \exp\left(-\frac{\epsilon \ell}{2b}\right)M(n, m, \ell)\Phi(b, \mathscr{L}, \ell),
\end{equation}
where
\begin{equation}\label{eq:phi}
    \Phi(b, \mathscr{L}, \ell) = \sum_{k=0}^N\frac{(-1)^k}{k+1}\sum_{\substack{G\subseteq I\\|G| = k\\ i(\ell)\not\in G}} \prod_{j\in G} \exp\left(-\frac{\epsilon \mathscr{L}_j}{2b}\right),
\end{equation}
then (ii) building a modified Hamming distance NFA~$A_{w, b} = (Q_{w, b}, \Sigma, q_{0,0}, \delta, \{q_{n, \ell}\})$, and finally (iii) using Algorithm~\ref{algo:MNFA} to synthesize a policy. A private output word~$w'= \sigma_1'\cdots\sigma_n'\in\mathcal{L}(A_{w, b})$ is generated by running~$A_{w, b}$ once. 
\end{mechanism}
 
The selection of the smallest index~$j$ for~$i(\ell)$ is arbitrary because using any index~$j$ where~$\mathscr{L}_j = \ell$ yields an identical value of~$\Phi(b, \mathscr{L}, \ell)$. We use the smallest 
such index 
for concreteness. 
Also, we observe 
that evaluating~$\Phi(b, \mathscr{L}, \ell)$ from~\eqref{eq:phi} is equivalent to evaluating~$\Psi(b, L, \hat{w}(\ell))$ from~\eqref{eq:psi}, where~$\hat{w}(\ell)$ is an arbitrary word such that~$d(w, \hat{w}(\ell)) = \ell$. This equivalence allows us to compute the probability of selecting a private output word with error~$\ell$ by operating only on the vector of Hamming distances~$\mathscr{L}$, and not evaluating the Hamming distance between a specific private output word and the sensitive input word. As a result, we may implement the permute-and-flip mechanism without explicitly evaluating any of the~$m^n$ Hamming distances between the sensitive input word and all possible private output words. 

\begin{theorem}\label{thm:mech1_priv}
    Fix a probability space~$(\Omega, \mathcal{F}, \mathbb{P})$. Given an adjacency parameter~$b\in\mathbb{N}^+$, a privacy parameter~$\epsilon\geq 0$, and a sensitive word~$w\in\Sigma^n$, Mechanism~\ref{mech:prob1} provides word~$\epsilon$-differential privacy to~$w$. 
\end{theorem}
    \noindent\emph{Proof.} See Appendix~\ref{prf:thm1}. \qed

Mechanism~\ref{mech:prob1} allows for more efficient sampling of private output words than
a direct implementation of 
Definition~\ref{def:PF}.
This improvement in efficiency is attained  
because we restrict the set of possible private output words to only those with~$\ell$ errors,  without needing to enumerate those words \emph{a priori}. 

\begin{figure*}[!t]
    \centering
    \begin{subfigure}{0.3\linewidth}
        \centering
%
%
\definecolor{chocolate2267451}{RGB}{226,74,51}
\definecolor{dimgray85}{RGB}{85,85,85}
\definecolor{gainsboro229}{RGB}{229,229,229}
\definecolor{lightgray204}{RGB}{204,204,204}
\definecolor{steelblue52138189}{RGB}{52,138,189}
\definecolor{black}{RGB}{0, 0, 0}
\definecolor{GTblue}{RGB}{0, 48, 87}
\definecolor{GTgold}{RGB}{179, 163, 105}
\definecolor{UFOrange}{RGB}{250, 70, 22}
\definecolor{UFblue}{RGB}{0, 33, 165}
\begin{tikzpicture}

\begin{axis}[%
width=0.25\figW,
height=0.65\figH,
axis background/.style={fill=gainsboro229},
axis line style={white},
scale only axis,
xlabel=\textcolor{black}{Privacy Strength, $\epsilon$},
xtick style={color=dimgray85},
x grid style={white},
yminorticks=true,
y grid style={white},
ylabel=\textcolor{black}{$\E{\ell}$},
xmajorgrids,
ymajorgrids,
yminorgrids,
tick align=outside,
tick pos=left,
legend pos = north east,
legend style={nodes={scale=0.75, transform shape}},
]
\addplot [color=UFblue, ultra thick]
  table[row sep=crcr]{%
0.1	2.43549276444719\\
0.302040816326531	2.3053788292915\\
0.504081632653061	2.17590860762373\\
0.706122448979592	2.04747713148513\\
0.908163265306122	1.92042592183956\\
1.11020408163265	1.79502288689028\\
1.31224489795918	1.67145578792973\\
1.51428571428571	1.54987318509186\\
1.71632653061225	1.43049584659039\\
1.91836734693878	1.31375364892068\\
2.12040816326531	1.20035282272755\\
2.32244897959184	1.09121399568818\\
2.52448979591837	0.987326174388964\\
2.7265306122449	0.889581106065586\\
2.92857142857143	0.79864885103903\\
3.13061224489796	0.714927953592282\\
3.33265306122449	0.638539799949239\\
3.53469387755102	0.569369223173893\\
3.73673469387755	0.507120530406608\\
3.93877551020408	0.451373142345834\\
4.14081632653061	0.401634985463621\\
4.34285714285714	0.357378966446369\\
4.54489795918367	0.318041232594676\\
4.7469387755102	0.283216745179981\\
4.94897959183673	0.252318036691519\\
5.15102040816326	0.2242555901914\\
5.3530612244898	0.200670357618772\\
5.55510204081633	0.179155980362891\\
5.75714285714286	0.160070531632939\\
5.95918367346939	0.143127807058992\\
6.16122448979592	0.128074682269675\\
6.36326530612245	0.114688541822153\\
6.56530612244898	0.102773692442789\\
6.76734693877551	0.0921582772533679\\
6.96938775510204	0.0826915932224887\\
7.17142857142857	0.0742417791997859\\
7.3734693877551	0.0666926211643125\\
7.57551020408163	0.0599424444108034\\
7.77755102040816	0.0539015961621792\\
7.97959183673469	0.048491311037979\\
8.18163265306122	0.0436421685766874\\
8.38367346938776	0.0392930346848112\\
8.58571428571429	0.0353897477967412\\
8.78775510204082	0.0318844888587031\\
8.98979591836735	0.0287349722779944\\
9.19183673469388	0.0259035433307559\\
9.39387755102041	0.0233569546247759\\
9.59591836734694	0.0210654587761842\\
9.79795918367347	0.0190026109258217\\
10	0.01714519434789\\
};
\addlegendentry{Mechanism~\ref{mech:prob1}}
\addplot [color=red, dashed, ultra thick]
  table[row sep=crcr]{%
0.1	2.43751301757895\\
0.302040816326531	2.31158245955396\\
0.504081632653061	2.18660624282469\\
0.706122448979592	2.06320139740531\\
0.908163265306122	1.94195380353289\\
1.11020408163265	1.82340754312836\\
1.31224489795918	1.70805593369655\\
1.51428571428571	1.59633452006134\\
1.71632653061225	1.48861617480401\\
1.91836734693878	1.38520833406422\\
2.12040816326531	1.28635228181919\\
2.32244897959184	1.19222430091279\\
2.52448979591837	1.10293843806486\\
2.7265306122449	1.01855058489812\\
2.92857142857143	0.939063556958632\\
3.13061224489796	0.864432854902975\\
3.33265306122449	0.794572812190433\\
3.53469387755102	0.729362866800179\\
3.73673469387755	0.66865373572559\\
3.93877551020408	0.612273315793993\\
4.14081632653061	0.560032179016854\\
4.34285714285714	0.511728572412457\\
4.54489795918367	0.467152869203846\\
4.7469387755102	0.426091449452082\\
4.94897959183673	0.388330013193157\\
5.15102040816326	0.353656348172716\\
5.3530612244898	0.321862587830824\\
5.55510204081633	0.292747003998838\\
5.75714285714286	0.266115383640233\\
5.95918367346939	0.241782040709096\\
6.16122448979592	0.219570513572002\\
6.36326530612245	0.199313996110941\\
6.56530612244898	0.180855547162542\\
6.76734693877551	0.164048118810725\\
6.96938775510204	0.148754439593552\\
7.17142857142857	0.134846784177729\\
7.3734693877551	0.122206656686465\\
7.57551020408163	0.110724410765633\\
7.77755102040816	0.100298825717227\\
7.97959183673469	0.0908366546576036\\
8.18163265306122	0.0822521576828112\\
8.38367346938776	0.0744666304364097\\
8.58571428571429	0.0674079362557486\\
8.78775510204082	0.0610100481915544\\
8.98979591836735	0.0552126056198254\\
9.19183673469388	0.0499604888597477\\
9.39387755102041	0.045203414142513\\
9.59591836734694	0.0408955504110743\\
9.79795918367347	0.036995158740084\\
10	0.0334642546214243\\
};
\addlegendentry{Theorem~\ref{thm:accuracy_bounds} (UB)}
\addplot [color=UFOrange, dotted, ultra thick]
  table[row sep=crcr]{%
0.1	2.3933185068274\\
0.302040816326531	2.25540739338091\\
0.504081632653061	2.11610243253527\\
0.706122448979592	1.97581761975735\\
0.908163265306122	1.83497931951714\\
1.11020408163265	1.69402744006015\\
1.31224489795918	1.55341579566419\\
1.51428571428571	1.4136110469092\\
1.71632653061225	1.2750898069242\\
1.91836734693878	1.13833376609807\\
2.12040816326531	1.00382297157621\\
2.32244897959184	0.872027653138985\\
2.52448979591837	0.743399175486708\\
2.7265306122449	0.618360796181374\\
2.92857142857143	0.497298914032214\\
3.13061224489796	0.380555415842113\\
3.33265306122449	0.268421592323439\\
3.53469387755102	0.161133923736474\\
3.73673469387755	0.058871858891152\\
3.93877551020408	-0.0382424505912745\\
4.14081632653061	-0.130142630467829\\
4.34285714285714	-0.216815042029824\\
4.54489795918367	-0.29829453650688\\
4.7469387755102	-0.37465934331738\\
4.94897959183673	-0.446025420198367\\
5.15102040816326	-0.512540562653271\\
5.3530612244898	-0.574378526345447\\
5.55510204081633	-0.631733365927596\\
5.75714285714286	-0.684814142825706\\
5.95918367346939	-0.733840106651852\\
6.16122448979592	-0.779036412721907\\
6.36326530612245	-0.820630402901934\\
6.56530612244898	-0.858848449060452\\
6.76734693877551	-0.893913337461063\\
6.96938775510204	-0.926042157773506\\
7.17142857142857	-0.955444651077357\\
7.3734693877551	-0.982321966277622\\
7.57551020408163	-1.0068657727738\\
7.77755102040816	-1.02925767814423\\
7.97959183673469	-1.0496689022713\\
8.18163265306122	-1.06826016312001\\
8.38367346938776	-1.08518173380305\\
8.58571428571429	-1.10057363524748\\
8.78775510204082	-1.11456593345207\\
8.98979591836735	-1.12727911480281\\
9.19183673469388	-1.13882451707724\\
9.39387755102041	-1.14930479754813\\
9.59591836734694	-1.15881442296165\\
9.79795918367347	-1.16744016911204\\
10	-1.17526162027871\\
};
\addlegendentry{Theorem~\ref{thm:accuracy_bounds} (LB)}

\end{axis}
\end{tikzpicture}%
        \caption{}
    \end{subfigure}
    \hfill 
    \begin{subfigure}{0.3\textwidth}
        \centering
%
%
\definecolor{chocolate2267451}{RGB}{226,74,51}
\definecolor{dimgray85}{RGB}{85,85,85}
\definecolor{gainsboro229}{RGB}{229,229,229}
\definecolor{lightgray204}{RGB}{204,204,204}
\definecolor{steelblue52138189}{RGB}{52,138,189}
\definecolor{black}{RGB}{0, 0, 0}
\definecolor{GTblue}{RGB}{0, 48, 87}
\definecolor{GTgold}{RGB}{179, 163, 105}
\definecolor{UFOrange}{RGB}{250, 70, 22}
\definecolor{UFblue}{RGB}{0, 33, 165}
\begin{tikzpicture}

\begin{axis}[%
width=0.25\figW,
height=0.65\figH,
axis background/.style={fill=gainsboro229},
axis line style={white},
scale only axis,
xlabel=\textcolor{black}{Word length, $n$},
xtick style={color=dimgray85},
x grid style={white},
yminorticks=true,
y grid style={white},
xmajorgrids,
ymajorgrids,
yminorgrids,
tick align=outside,
tick pos=left,
legend pos = north east,
]
\addplot [color=UFblue, ultra thick]
  table[row sep=crcr]{%
1	0.0410424666030647\\
2	0.0860389860348599\\
3	0.135044105171023\\
4	0.188071146445207\\
5	0.245086526837725\\
6	0.306014809476947\\
7	0.370720319617681\\
8	0.439022018004397\\
9	0.51068931279126\\
10	0.585443773232246\\
11	0.662967010775305\\
12	0.742906647720533\\
13	0.824886873636226\\
14	0.908516702982929\\
15	0.993407852864674\\
16	1.07917948451989\\
17	1.16547422043473\\
18	1.25197076830195\\
19	1.33838489183453\\
20	1.42448334141638\\
};
\addplot [color=red, dashed, line width=2.0pt]
  table[row sep=crcr]{%
1	0.0758581800212436\\
2	0.151716360042487\\
3	0.227574540063731\\
4	0.303432720084974\\
5	0.379290900106218\\
6	0.455149080127461\\
7	0.531007260148705\\
8	0.606865440169948\\
9	0.682723620191192\\
10	0.758581800212436\\
11	0.834439980233679\\
12	0.910298160254923\\
13	0.986156340276166\\
14	1.06201452029741\\
15	1.13787270031865\\
16	1.2137308803399\\
17	1.28958906036114\\
18	1.36544724038238\\
19	1.44130542040363\\
20	1.51716360042487\\
};
\addplot [color=UFOrange, ultra thick, dotted]
  table[row sep=crcr]{%
1	-0.155177274973446\\
2	-0.275302691674337\\
3	-0.364364705365052\\
4	-0.425948362031569\\
5	-0.463273550750353\\
6	-0.47922977394933\\
7	-0.476407743733518\\
8	-0.45712808572905\\
9	-0.423467407462462\\
10	-0.377281965943141\\
11	-0.320229148679554\\
12	-0.253786963672889\\
13	-0.179271716847082\\
14	-0.0978540397563274\\
15	-0.0105734161358471\\
16	0.0816486571849135\\
17	0.177996751409595\\
18	0.277750716238875\\
19	0.380275916918936\\
20	0.485014404212767\\
};

\end{axis}
\end{tikzpicture}%
        \caption{}
    \end{subfigure}
    \hfill
    \begin{subfigure}{0.3\textwidth}
        \centering
%
%
\definecolor{chocolate2267451}{RGB}{226,74,51}
\definecolor{dimgray85}{RGB}{85,85,85}
\definecolor{gainsboro229}{RGB}{229,229,229}
\definecolor{lightgray204}{RGB}{204,204,204}
\definecolor{steelblue52138189}{RGB}{52,138,189}
\definecolor{black}{RGB}{0, 0, 0}
\definecolor{GTblue}{RGB}{0, 48, 87}
\definecolor{GTgold}{RGB}{179, 163, 105}
\definecolor{UFOrange}{RGB}{250, 70, 22}
\definecolor{UFblue}{RGB}{0, 33, 165}
\begin{tikzpicture}

\begin{axis}[%
width=0.25\figW,
height=0.65\figH,
axis background/.style={fill=gainsboro229},
axis line style={white},
scale only axis,
xlabel=\textcolor{black}{Alphabet size, $m$},
xtick style={color=dimgray85},
x grid style={white},
yminorticks=true,
y grid style={white},
xmajorgrids,
ymajorgrids,
yminorgrids,
tick align=outside,
tick pos=left,
legend pos = north east,
]
\addplot [color=UFblue, ultra thick]
  table[row sep=crcr]{%
2	0.245086526837725\\
3	0.538508919567994\\
4	0.844981604448407\\
5	1.13376472679817\\
6	1.38883005039343\\
7	1.608675798611\\
8	1.79869191837619\\
9	1.96532349866505\\
10	2.1135136011154\\
};
\addplot [color=red, dashed, line width=2.0pt]
  table[row sep=crcr]{%
2	0.379290900106218\\
3	0.705094606612051\\
4	0.987979974754094\\
5	1.23590344285416\\
6	1.45496922956442\\
7	1.64993867534546\\
8	1.82458026947381\\
9	1.98191561845045\\
10	2.12439574317743\\
};
\addplot [color=UFOrange, ultra thick, dotted]
  table[row sep=crcr]{%
2	-0.463273550750353\\
3	0.120536079200013\\
4	0.57218862337674\\
5	0.93365521112927\\
6	1.2310146038355\\
7	1.48115758939824\\
8	1.69543881122727\\
9	1.8817504142339\\
10	2.0457445020015\\
};

\end{axis}
\end{tikzpicture}%
        \caption{}
    \end{subfigure}
    \caption{The upper and lower bounds for the expected error~$\E{\ell}$ from Theorem~\ref{thm:accuracy_bounds} with (a) varying~$\epsilon\in[0.1, 10]$ and fixed~$n = 5$ and~$m = 2$, (b) varying~$n\in\{2, 3,\ldots, 20\}$ 
    and fixed~$\epsilon = 5$ and~$m = 2$, and (c) varying~$m\in\{2, 3,\ldots, 10\}$ and fixed~$\epsilon = 5$ and~$n = 5$. The bounds are tightest with large~$m$ and~$n$, and small~$\epsilon$.}
    \label{fig:bound_comparisons}
\end{figure*}
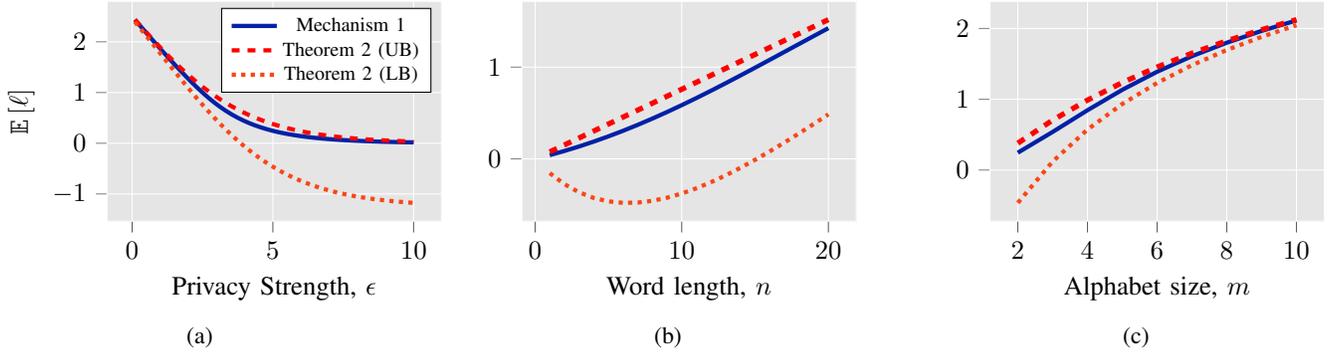

\subsection{Mechanism Accuracy}
In this subsection, we solve Problem~\ref{prob:accuracy}. To quantify the accuracy of the mechanism, we develop bounds on the expected number of errors in a private output word~$w'$ as a function of the adjacency parameter~$b$, the privacy parameter~$\epsilon$, the word length~$n$, and the alphabet size~$m$.

\begin{theorem}[Solution to Problem~\ref{prob:accuracy}]\label{thm:accuracy_bounds}
    Consider Mechanism~\ref{mech:prob1}. 
    Then the expected error~$\E{\ell}$ is bounded according to 
    \begin{equation}
        \frac{nC}{1+C} -\frac{n(\Phi(b,\mathscr{L}, 0)-\Phi(b,\mathscr{L}, n))}{4\Eem{\Phi(b, \mathscr{L}, \ell)}}\leq \E{\ell} \leq \frac{nC}{1+C}, 
    \end{equation}
    where~$\Phi$ is from~\eqref{eq:phi},
    $C = (m-1)e^{-\frac{\epsilon}{2b}}$, 
    and~$\Eem{\cdot}$ is the expectation under the distribution defined by 
    \begin{equation}
        \probem{\ell} = \frac{1}{Z}\exp\left(-\frac{\epsilon\ell}{2b}\right)\binom{n}{\ell}(m-1)^{\ell}.
    \end{equation}
    For all~$t>0$, we have 
        $\prob{|\ell-\E{\ell}|\geq t}\leq 2\exp(-\frac{2t^2}{n^2}).$
\end{theorem}
 \noindent\emph{Proof.}   See Appendix~\ref{prf:thm2}. \qed

The upper bound in Theorem~\ref{thm:accuracy_bounds} exactly matches the expression for the expected
error for the exponential mechanism in~\cite{chen2023differentialsymbolic}. Previously, 
\cite[Theorem 2]{mckenna2020permute} proved that the permute-and-flip mechanism is never worse than the exponential mechanism,
and we have recovered that fact in explicit form 
in
the upper bound in Theorem~\ref{thm:accuracy_bounds}. 


Figure~\ref{fig:bound_comparisons} illustrates the bounds in Theorem~\ref{thm:accuracy_bounds} with varying~$\epsilon$, word length~$n$, and alphabet size~$m$. 
As~$n$ and~$m$ increase, the error from Mechanism~\ref{mech:prob1} more closely resembles that of~\cite[Mechanism 1]{chen2023differentialsymbolic}. 

\section{Extension to Markov Chains}\label{sec:MC}
In this section, we solve Problem~\ref{prob:markov_chains}, and we modify Mechanism~\ref{mech:prob1} to ensure private output words are feasible trajectories for a given Markov chain. 
We do this by selecting a private output word using a product modified Hamming distance NFA.

\begin{definition}[Product Modified Hamming Distance NFA; \cite{chen2023differentialsymbolic}]\label{def:PMNFA}
    Let a Markov chain~$(\mathcal{Y}, T, y_0$) be given.    
    For a sequence of states~$x\in\mathcal{Y}^n$ and a distance~$\ell \in \mathbb{N}$, let~$A_{x, \ell} = (Q_{x, \ell}, \Sigma, q^0, \delta, F_{n, \ell})$ be a MNFA. Then the Product Modified Hamming Distance NFA (P-MNFA) is an MNFA~$A_{x, \ell, \mathcal{Y}} = (Q_{\mathcal{Y}}, \Sigma, q_{\mathcal{Y}}^0, \delta_{\mathcal{Y}}, F_{\mathcal{Y}})$, where
    \begin{multline}
       \!\!\!\!\! Q_{\mathcal{Y}} = Q_{x, \ell}\times \mathcal{Y}, \quad \! \delta_{\mathcal{Y}}: Q_{x, \ell}\times \mathcal{Y} \times \Sigma \to 2^{Q_{\mathcal{Y}}}, \quad \!\! q_{\mathcal{Y}}^0 = (q_0, y_0),\\ \text{and}~ F_{\mathcal{Y}} = \{(q_f, y)\in Q_{\mathcal{Y}}\mid q_f\in F_{n, \ell}, y\in\mathcal{Y}\},
    \end{multline}
    and for any~$(q', y')\in\delta_{\mathcal{Y}}(q, y, \sigma)$, we have~$\delta(q, \sigma) = q'$ and~$\prob{y'\mid y}>0$. A state~$q_y\in Q_{\mathcal{Y}}$ can transition to another state by a policy~$\mu_y(\cdot, \cdot \mid q_y, y):Q_{x,\ell}\times \mathcal{Y}\to [0, 1]$. And~$\mathcal{L}(A_{x, \ell, \mathcal{Y}})$ is the set of all feasible words of length~$n$ with Hamming from~$x$ equal to~$\ell$. 
\end{definition}
\begin{algorithm}[t]
    \caption{Constructing P-MNFA and Policy Synthesis for Markov chains~\cite{chen2023differentialsymbolic}}
    \label{algo:PMNFA}
    \begin{algorithmic}[1]
    \STATEx \textbf{Inputs}: Sensitive word $w$ with length~$n$, transition function $\delta$, accepting set $\{(q_{n, \ell}, y)\mid y\in \mathcal{Y}\}$
    \STATEx \textbf{Outputs}: Policy~$\mu_y$
    \vspace{1mm}
    \STATE $V(q_{n, \ell}, y) = 1~\text{for all } y\in \mathcal{Y}$ 
    \STATE $CurrQ = \{(q_{n, \ell}, y)\}~\text{for all } y\in \mathcal{Y}$
    \STATE $ActiveQ = \{\}$
    \FOR{$iteration\in [n]$}
    \FOR{$(q', y')\in CurrQ$}
    \FOR{$(q, y)~s.t.~(q', y')\in \delta_{\mathcal{Y}}(q, y, y')$} 
         \STATE $V(q, y) = \sum_{(q'', y'')\mid (q'', y'')\in \delta_{\mathcal{Y}}(q, y, y'') } V(q'', y'')$
         \STATE $\mu_y(q', y' \mid q, y) = \frac{V(q', y')}{V(q, y)}$
         \STATE $ActiveQ = ActiveQ\cup\{(q, y)\}$
         \ENDFOR
    \ENDFOR
    \STATE $CurrQ = ActiveQ$
    \STATE $ActiveQ = \{\}$
    \ENDFOR
    \end{algorithmic}
\end{algorithm}
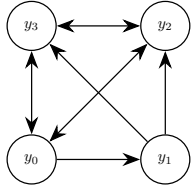
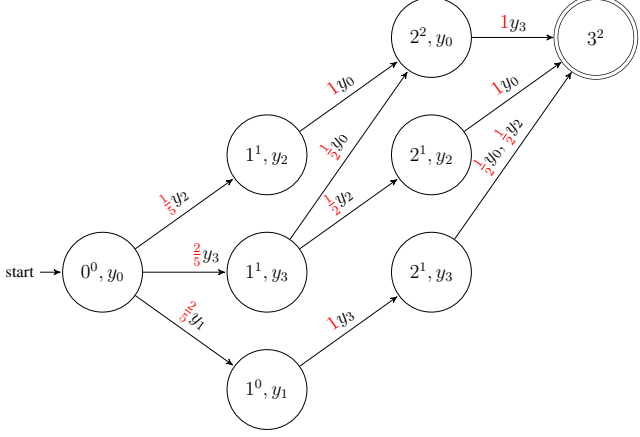
\begin{figure}
    \centering
    \begin{subfigure}{\linewidth}
        \centering
        \resizebox{0.3\linewidth}{!}{
\begin{tikzpicture}[
    >={Stealth[length=4.0mm]}, 
    auto,
    thick, 
    node distance=3cm, 
    every state/.style={
        draw, 
        circle, 
        minimum size=1.1cm, 
        inner sep=0pt
    }
]

    
    \node[state] (s3) at (0,3) {$y_3$};
    \node[state] (s2) at (3,3) {$y_2$};
    \node[state] (s0) at (0,0) {$y_0$};
    \node[state] (s1) at (3,0) {$y_1$};

    
    \draw[<->] (s3) -- (s2);
    
    \draw[<->] (s0) -- (s3);
    
    \draw[<->] (s0) -- (s2);
    
    \draw[->] (s1) -- (s3);

    
    \draw[->] (s0) -- (s1);
    
    \draw[->] (s1) -- (s2);

\end{tikzpicture}
}
        \caption{Markov chain on the state space~$\{y_0, y_1, y_2, y_3\}$. Directed edges represent feasible transitions.
        }
    \end{subfigure}
    \hfill
    \begin{subfigure}{\linewidth}
        \centering
        \resizebox{\linewidth}{!}{
\begin{tikzpicture}[
    ->, 
    >={Stealth[round]}, 
    shorten >=1pt, 
    auto, 
    on grid, 
    node distance=3.5cm and 2.5cm, 
    every state/.style={
        draw,
        circle,
        minimum size=1.7cm, 
        inner sep=1pt, 
        font=\large 
    },
    initial text={start}, 
    edge label/.style={
        font=\large, 
        sloped, 
        text=black 
    }
]


    \node[state, initial] (q00) at (0,0) {$0^0, y_0$};

    \node[state] (q11_s2) at (3.5, 2.5) {$1^1, y_2$};
    \node[state] (q11_s3) at (3.5, 0)   {$1^1, y_3$};
    \node[state] (q10_s1) at (3.5, -2.5) {$1^0, y_1$};

    \node[state] (q22_s0) at (7, 5)   {$2^2, y_0$};
    \node[state] (q21_s2) at (7, 2.5) {$2^1, y_2$};
    \node[state] (q21_s3) at (7, 0)   {$2^1, y_3$};

    \node[state, accepting] (q32) at (10.5, 5) {$3^2$};


    \path
        (q00) edge node[edge label] {$\textcolor{red}{\frac{1}{5}}y_2$} (q11_s2)
        (q00) edge node[edge label, near end] {$\textcolor{red}{\frac{2}{5}}y_3$} (q11_s3)
        (q00) edge node[edge label] {$\textcolor{red}{\frac{2}{5}}y_1$} (q10_s1)

        (q11_s2) edge node[edge label] {$\textcolor{red}{1}y_0$} (q22_s0)
        (q11_s3) edge node[edge label] {$\textcolor{red}{\frac{1}{2}}y_0$} (q22_s0)
        (q11_s3) edge node[edge label] {$\textcolor{red}{\frac{1}{2}}y_2$} (q21_s2)
        (q10_s1) edge node[edge label] {$\textcolor{red}{1}y_3$} (q21_s3)

        (q22_s0) edge node[edge label] {$\textcolor{red}{1}y_3$} (q32)
        (q21_s2) edge node[edge label] {$\textcolor{red}{1}y_0$} (q32)
        (q21_s3) edge
            node[edge label] {$\textcolor{red}{\frac{1}{2}}y_0, \textcolor{red}{\frac{1}{2}}y_2$}
        (q32);

\end{tikzpicture}
}
        \caption{Product MNFA for the Markov chain in (a) for a trajectory of length~$n = 3$ 
        and a Hamming distance of~$\ell = 2$ for the word~$y_1y_2y_3$, with an initial state~$y_0$.         
        Each circle represents a state in the MNFA      
        where the base of the number in the circle is the current length of the word and the superscript is the number of errors. The Markov chain state~$y_i$ is the symbol most recently appended
        to the output word that is being constructed.        
        }
    \end{subfigure}
    \caption{A P-MNFA (b) for a four-state Markov chain (a).}
    \label{fig:PMNFA_example}
\end{figure}

Algorithm~\ref{algo:PMNFA} is used to construct a P-MNFA and synthesize the policy~$\mu_y$. Figure~\ref{fig:PMNFA_example} illustrates Definition~\ref{def:PMNFA} and Algorithm~\ref{algo:PMNFA} with an example Markov chain. 
Definition~\ref{def:PMNFA} implements the synchronous product of an MNFA and a Markov chain, 
and this construction ensures
that all words generated by running~$A_{x, \ell, \mathcal{Y}}$ are feasible in both the MNFA~$A_{x, \ell}$ and the Markov chain~$(\mathcal{Y}, T, y_0)$. Similar to Mechanism~\ref{mech:prob1}, to formally state our mechanism for Markov chains we define~$N_{\mathcal{Y}^n}(\ell)$ as the number of length~$n$ words in~$\mathcal{Y}^n$
with initial state~$y_0$ and Hamming distance~$\ell$ from the sensitive input word. Then we define 
    $\mathscr{L}^{\mathcal{Y}^n} = \left[0, \ones^{N_{\mathcal{Y}^n}(1)}, \cdots, n \ones^{N_{\mathcal{Y}^n}(n)}\right]^T\in \mathbb{R}^{|\mathcal{L}(\mathcal{Y}^n)|}.$
The mechanism itself is as follows. 

\begin{mechanism}[Solution to Problem~\ref{prob:markov_chains}]\label{mech:prob3}
    Fix a probability space~$(\Omega, \mathcal{F}, \mathbb{P})$, a Markov chain~$(\mathcal{Y}, T, y_0)$, and an adjacency parameter~$b\!\in\!\mathbb{N}^+$. Let a sensitive input word~$w\! =\! \!y_1\!\cdots\! y_n\!\in\!\mathcal{Y}^n$ be given. The mechanism~$\mathcal{M}_{2}: \mathcal{Y}^n\times \Omega \to \mathcal{Y}^n$ selects a private output word by (i) drawing a Hamming distance~$\ell$ from
    \begin{equation}\label{eq:mech2_prob_ell}
    \prob{\ell; w, b} = \exp\left(-\frac{\epsilon \ell}{2b}\right)N_{\mathcal{Y}^n}(\ell)\Phi(b, \mathscr{L}^{\mathcal{Y}^n}, \ell),
\end{equation}
then (ii) constructing a P-MNFA~$A_{w, \ell, \mathcal{Y}} = (Q_{\mathcal{Y}}, \Sigma, q_{\mathcal{Y}}^0, \delta_{\mathcal{Y}}, F_{\mathcal{Y}})$ and synthesizing a policy~$\mu_{y}$ using~Algorithm~\ref{algo:PMNFA}. A private output word~$w' = s_1'\cdots s_n'\in\mathcal{L}(A_{w, \ell, \mathcal{Y}})$ is generated by running the P-MNFA~$A_{w, \ell, \mathcal{Y}}$ once. 
\end{mechanism}
The Markov chain~$(\mathcal{Y}, T, y_0)$ fixes the initial condition~$y_0$ by definition,
        and private state trajectories for this Markov chain must keep this initial condition to remain
        valid.
        Therefore, privacy does not alter~$y_0$ when we use it below, and we often treat
        a state trajectory as consisting only of the states that come after
        the initial state, e.g., we treat~$y_0y_1y_2y_3$ as~$y_1y_2y_3$ precisely
        because privacy cannot alter~$y_0$. 
The following theorem confirms that Mechanism~\ref{mech:prob3} 
provides~word $\epsilon$-differential privacy.

\begin{theorem}\label{thm:mech2}
    Fix a privacy parameter~$\epsilon\geq 0$ and an adjacency parameter~$b\in\mathbb{N}^+$. Let~$w\in\mathcal{Y}^n$ be a state sequence generated by a Markov chain~$(\mathcal{Y}, T, y_0)$. 
    Then Mechanism~\ref{mech:prob3} provides word~$\epsilon$-differential privacy to~$w$. 
\end{theorem}
 \noindent\emph{Proof.}   See Appendix~\ref{prf:thm3}.\qed

Mechanism~\ref{mech:prob3} restricts the output space of Mechanism~\ref{mech:prob1} to~$\mathcal{L}(\mathcal{Y}^n)$, i.e., only the state sequences that are feasible 
in the given Markov chain 
from the initial state~$y_0$
can be generated as private output words. 
\section{Numerical Simulations}\label{sec:sims}
In this section we solve Problem~\ref{prob:sims}. We compare the accuracy of Mechanism~\ref{mech:prob3} with 
that of~\cite[Mechanism 3]{chen2023differentialsymbolic} using the Gainesville, Florida Annual Average Daily Traffic (AADT)~\cite{GainesvilleTraffic} data that was also used in~\cite{chen2023differentialsymbolic,fallin2023differential}. We first generate a Markov chain using the traffic data in~\cite{GainesvilleTraffic} and then privatize state trajectories, i.e., words, produced by this Markov chain. Such words can represent, e.g., trips to acquaintances' homes, which may be sensitive. Therefore, we provide word differential privacy to these trajectories. 

To develop the Markov chain model, we divide the roads around the University of Florida into segments, where a road segment is a section of road between two intersections. If a road does not extend past any intersections, it has only a single segment. 
Each of these segments is a state in the Markov chain, and
the resulting Markov chain has~$|\mathcal{Y}| = 43$ states. To compute the transition probabilities, we count the number of times drivers transitioned from one road segment to another, and we divide it by the total number of times drivers transitioned away from the first segment. 

To analyze the expected error, we fix a sensitive input word of length~$n = 14$ shown in Figure~\ref{fig:trajectories}.  Figure~\ref{fig:trajectories} also illustrates example private output words for two values of~$\epsilon$. We find that under strong privacy, namely~$\epsilon = 0.5$, the sampled private output word in Figure~\ref{fig:trajectories}
differs from the sensitive input word in every state but the initial state, while at~$\epsilon = 5$ the sampled private output word is identical to the sensitive input word. 

Next,
we sample~$2,000$ private output words of length~$n = 14$ at~$\epsilon$ values in the range~$\epsilon \in [0.1, 10]$ and compute the average error in the trajectories at each~$\epsilon$. Figure~\ref{fig:error_comp} shows that these private output words incur error
close to their corresponding expectations. The value~$\epsilon = 0.5$ gives Mechanism~\ref{mech:prob3} 
an expected error of~$\E{\ell} = 12.7$. 
Similarly, for~$\epsilon = 5$, the expected error is~$\E{\ell} = 0.28$. 
When trajectories are selected using~\cite[Mechanism 3]{chen2023differentialsymbolic}, we find that for~$\epsilon = 0.5$ 
the expected error is~$\mathbb{E}_{Exp}[\ell] = 12.7$, while at~$\epsilon = 5$ the expected
error of that mechanism is~$\mathbb{E}_{Exp}[\ell] = 0.62$. 
This pattern agrees with 
the trend seen in Figure~\ref{fig:bound_comparisons}, where accuracy 
of the two mechanisms 
is roughly equal under strong privacy, but at~$\epsilon = 5$ we find up to a~$55.7\%$ reduction in error. 
For all~$\epsilon>3$, we find at least a~$25\%$ reduction in error relative to~\cite[Mechanism 3]{chen2023differentialsymbolic} when using  Mechanism~\ref{mech:prob3}, which shows
sustained improvement in the accuracy of our mechanism over the state of the art. 

\begin{figure}
    \centering
    \includegraphics[width=0.9\linewidth]{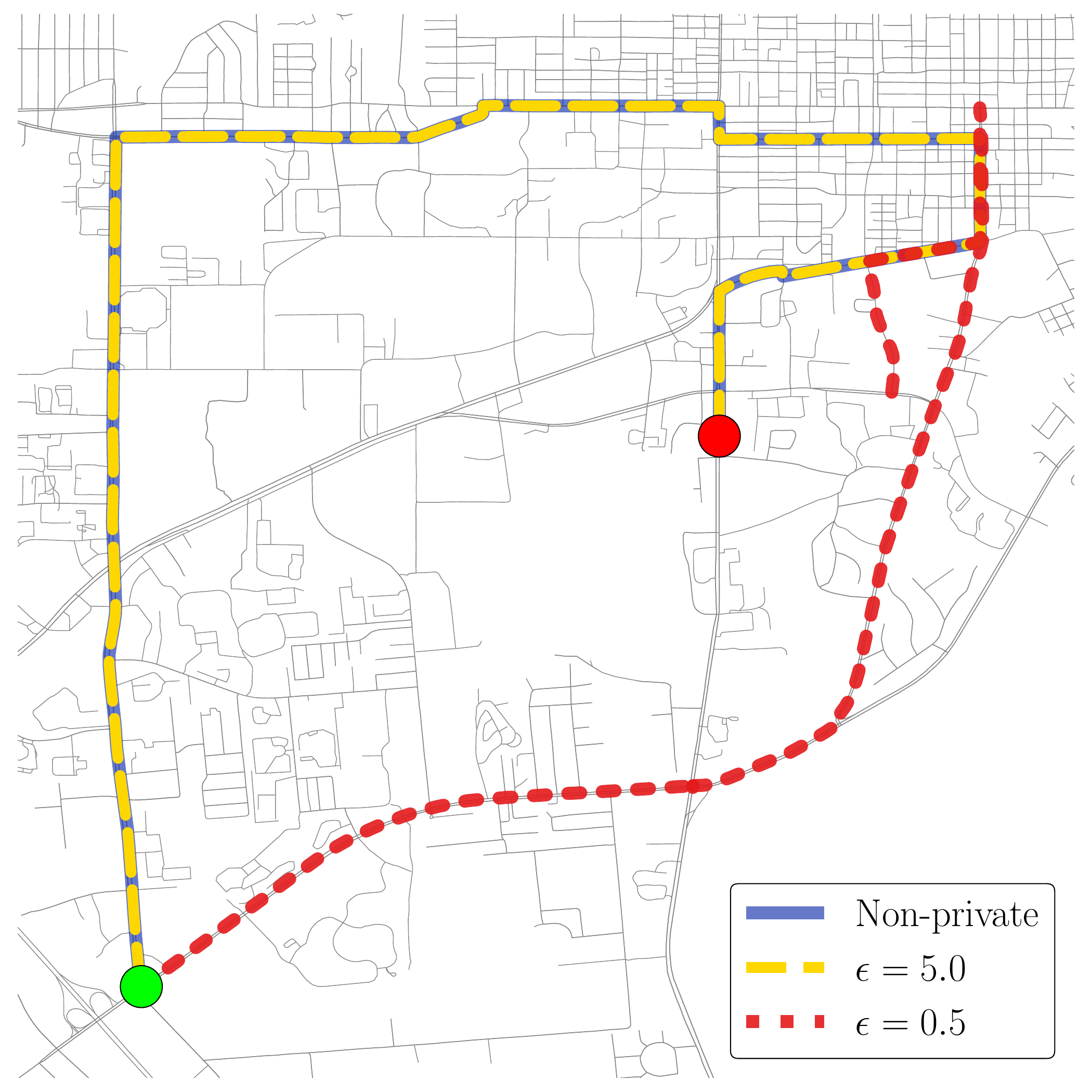}
    \caption{Sample private output words with~$n = 14$ through Gainesville, Florida with SW 34th St as the initial state (green point).     
    At~$\epsilon = 5$, the average error is less than~$1$, and private output words are often close to the sensitive input word. At stronger privacy, i.e.,~$\epsilon =0.5$, the average error approaches~$13$, and private output words often differ in every state from the input word, which is the case 
    for the sampled private output word with~$\epsilon = 0.5$ shown here.    
    }
    \label{fig:trajectories}
\end{figure}

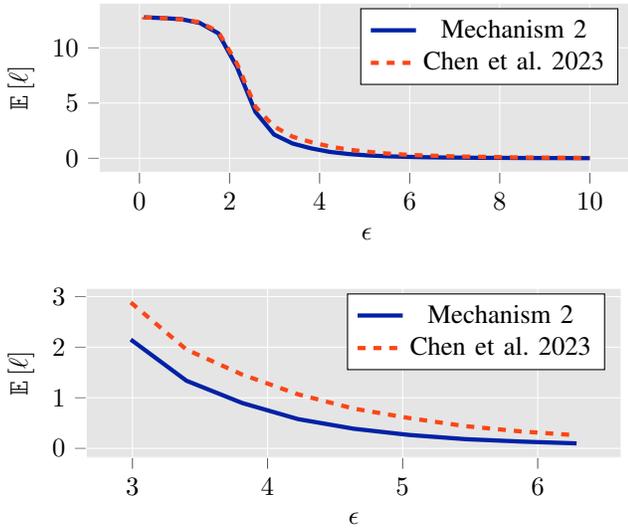
\begin{figure}
\captionsetup[subfigure]{justification=centering}
\begin{subfigure}{0.4\textwidth}
\centering
\hspace*{-1cm}
%
%
\definecolor{chocolate2267451}{RGB}{226,74,51}
\definecolor{dimgray85}{RGB}{85,85,85}
\definecolor{gainsboro229}{RGB}{229,229,229}
\definecolor{lightgray204}{RGB}{204,204,204}
\definecolor{steelblue52138189}{RGB}{52,138,189}
\definecolor{black}{RGB}{0, 0, 0}
\definecolor{GTblue}{RGB}{0, 48, 87}
\definecolor{GTgold}{RGB}{179, 163, 105}
\definecolor{UFOrange}{RGB}{250, 70, 22}
\definecolor{UFblue}{RGB}{0, 33, 165}
\begin{tikzpicture}

\begin{axis}[%
width=0.4\figW,
height=0.5\figH,
axis background/.style={fill=gainsboro229},
axis line style={white},
scale only axis,
xlabel=\textcolor{black}{$\epsilon$},
xtick style={color=dimgray85},
x grid style={white},
yminorticks=true,
y grid style={white},
ylabel=\textcolor{black}{$\E{\ell}$},
xmajorgrids,
ymajorgrids,
yminorgrids,
tick align=outside,
tick pos=left,
legend pos = north east,
]

\addplot [color=UFblue, ultra thick]
  table[row sep=crcr]{%
            0.1  12.763  \\
            0.5125  12.696  \\
            0.925  12.595  \\
            1.3375  12.2572  \\
            1.75  11.318  \\
            2.1625  8.2646  \\
            2.575  4.1926  \\
            2.9875  2.1514  \\
            3.4  1.3366  \\
            3.8125  0.8968  \\
            4.225  0.5788  \\
            4.6375  0.3882  \\
            5.05  0.2646  \\
            5.4625  0.1824  \\
            5.875  0.135  \\
            6.2875  0.0986  \\
            6.7  0.0728  \\
            7.1125  0.0598  \\
            7.525  0.0406  \\
            7.9375  0.0382  \\
            8.35  0.0288  \\
            8.7625  0.023  \\
            9.175  0.0194  \\
            9.5875  0.015  \\
            10.0  0.0104  \\
};
\addlegendentry{Mechanism~\ref{mech:prob3}}
\addplot [color=UFOrange, ultra thick, dashed]
  table[row sep=crcr]{%
            0.1  12.772  \\
            0.1  12.762  \\
            0.5125  12.691  \\
            0.925  12.57  \\
            1.3375  12.3304  \\
            1.75  11.4562  \\
            2.1625  8.5868  \\
            2.575  4.662  \\
            2.9875  2.8848  \\
            3.4  1.9526  \\
            3.8125  1.459  \\
            4.225  1.0694  \\
            4.6375  0.7862  \\
            5.05  0.5982  \\
            5.4625  0.44  \\
            5.875  0.3348  \\
            6.2875  0.2592  \\
            6.7  0.2114  \\
            7.1125  0.157  \\
            7.525  0.1368  \\
            7.9375  0.108  \\
            8.35  0.0792  \\
            8.7625  0.0638  \\
            9.175  0.0584  \\
            9.5875  0.042  \\
            10.0  0.0346  \\
};
\addlegendentry{Chen et al. 2023}



\end{axis}

\end{tikzpicture}
\end{subfigure}
\begin{subfigure}{0.4\textwidth}
\centering
\hspace*{-1cm}
    \definecolor{chocolate2267451}{RGB}{226,74,51}
\definecolor{dimgray85}{RGB}{85,85,85}
\definecolor{gainsboro229}{RGB}{229,229,229}
\definecolor{lightgray204}{RGB}{204,204,204}
\definecolor{steelblue52138189}{RGB}{52,138,189}
\definecolor{black}{RGB}{0, 0, 0}
\definecolor{GTblue}{RGB}{0, 48, 87}
\definecolor{GTgold}{RGB}{179, 163, 105}
\definecolor{UFOrange}{RGB}{250, 70, 22}
\definecolor{UFblue}{RGB}{0, 33, 165}
\begin{tikzpicture}

\begin{axis}[name=inset, xshift=0.17\textwidth,yshift=1.3cm,
width=0.4\figW,
height=0.5\figH,
axis background/.style={fill=gainsboro229},
axis line style={white},
scale only axis,
xlabel=\textcolor{black}{$\epsilon$},
xtick style={color=dimgray85},
x grid style={white},
yminorticks=true,
y grid style={white},
ylabel=\textcolor{black}{$\E{\ell}$},
xmajorgrids,
ymajorgrids,
yminorgrids,
tick align=outside,
tick pos=left,
legend pos = north east,
]

\addplot [color=UFblue, ultra thick]
  table[row sep=crcr]{%
            2.9875  2.1514  \\
            3.4  1.3366  \\
            3.8125  0.8968  \\
            4.225  0.5788  \\
            4.6375  0.3882  \\
            5.05  0.2646  \\
            5.4625  0.1824  \\
            5.875  0.135  \\
            6.2875  0.0986  \\
};
\addlegendentry{Mechanism~\ref{mech:prob3}}
\addplot [color=UFOrange, ultra thick, dashed]
  table[row sep=crcr]{%
            2.9875  2.8848  \\
            3.4  1.9526  \\
            3.8125  1.459  \\
            4.225  1.0694  \\
            4.6375  0.7862  \\
            5.05  0.5982  \\
            5.4625  0.44  \\
            5.875  0.3348  \\
            6.2875  0.2592  \\
};
\addlegendentry{Chen et al. 2023}
 \end{axis}
 \end{tikzpicture}
\end{subfigure}
       \caption{Error comparison between Mechanism~\ref{mech:prob3} in the current paper 
       and~\cite[Mechanism 3]{chen2023differentialsymbolic} with the initial state in Figure~\ref{fig:trajectories}       
       for privacy parameters in the range $\epsilon \in [0.1, 10]$ (top) and $\epsilon \in [3, 6]$ (bottom). Under strong privacy, e.g., $\epsilon = 0.1$, both mechanisms yield nearly identical average errors, but at more common privacy levels, e.g.,~$\epsilon = 3$, 
       Mechanism~\ref{mech:prob3} shows a~$25\%$ reduction in error relative to~\cite[Mechanism 3]{chen2023differentialsymbolic}. 
    }
   \label{fig:error_comp}
\end{figure}

\section{Conclusion}
We have presented a framework for privatizing symbolic trajectories based on the permute-and-flip mechanism, answering an open question and directly improving upon the state of the art. We have proved that our framework is, at worst, equivalent in expected error to the prior state-of-the-art mechanism, and empirically we find up to a~$55.7\%$ reduction in error at word~$5$-differential privacy. Future work will develop an online mechanism for implementing the permute-and-flip mechanism
in real time as trajectories are generated.

\appendix
\subsection{Proof of Theorem~1}\label{prf:thm1}
The permute-and-flip mechanism is word~$\epsilon$-differentially private~\cite[Theorem 1]{mckenna2020permute}, 
and we will show Mechanism~\ref{mech:prob1} is word~$\epsilon$-differentially private by showing it selects a private output word with the same probabilities as 
the permute-and-flip mechanism in 
Definition~\ref{def:PF}. 
From~\eqref{eq:pf_pmf} in Definition~\ref{def:PF}, the permute-and-flip mechanism selects the private output word~$w'$ with probability
        $\prob{\mathcal{M}_{PF}(w) = w'} = \exp\left(-\frac{\epsilon \ell_w}{2b}\right)\Psi(b, L, w'),$
    where~$\ell_w$ is the Hamming distance between~$w$ and~$w'$. From Mechanism~\ref{mech:prob1}, the probability of selecting the private output word~$w'$ is
    \begin{align}
        \prob{\mathcal{M}_1(w) = w'} &= \prob{\ell_w; w, b}\frac{1}{\binom{n}{\ell_w}(m-1)^{\ell_w}} \\ 
                                     &= \exp\left(-\frac{\epsilon \ell_w}{2b}\right)\Phi(b, \mathscr{L}, \ell),\label{eq:m1_prob}
    \end{align}
    which is from~\eqref{eq:mech1_pmf}.
    Because~$d(w, w') = d(w, \hat{w}(\ell))$, where~$\hat{w}(\ell)$ is an arbitrary word with~$\ell$ errors,     
    and because $\Psi(b, L, \hat{w}(\ell)) = \Phi(b, \mathscr{L}, \ell)$, we have $\Phi(b, \mathscr{L}, \ell) = \Psi(b, L, w')$. Then, we can  write~\eqref{eq:m1_prob} 
    as 
         $\prob{\mathcal{M}_1(w) = w'} = \exp\left(-\frac{\epsilon \ell_w}{2b}\right)\Psi(b, L, w').$
    As a result, we have~$\prob{\mathcal{M}_{PF}(w) = w'} = \prob{\mathcal{M}_1(w) = w'}$. Then Mechanism~\ref{mech:prob1} implements the permute-and-flip mechanism and hence is word~$\epsilon$-differentially private.\qed

\subsection{Proof of Theorem 2}\label{prf:thm2}
We begin by deriving the upper bound in the theorem statement. 
When using~\cite[Mechanism 3]{chen2023differentialsymbolic}, the probability of selecting a Hamming distance~$\ell$  is
        $\probem{\ell} = \frac{1}{Z}\exp\left(-\frac{\epsilon\ell}{2b}\right)\binom{n}{\ell}(m-1)^{\ell},$
    where~$\frac{1}{Z}$ is a normalization constant. Substituting this expression into the probability distribution     
    from~\eqref{eq:mech1_pmf} in Mechanism~\ref{mech:prob1}, we have
    \begin{equation}\label{eq:prob_ell_1}
        \prob{\ell} =  \probem{\ell} \Phi(b, \mathscr{L}, \ell)Z,
    \end{equation}
    where~$\prob{\ell}$ denotes the probability of selecting a Hamming distance~$\ell$ under Mechanism~\ref{mech:prob3}. 
    For~$\prob{\ell}$ to be a valid probability mass function, we require~$\sum_{\ell=1}^{n} \prob{\ell} = 1$. Then, 
    \begin{equation}\label{eq:norm_constant}
        Z = \frac{1}{\sum_{\ell=1}^n \probem{\ell} \Phi(b, \mathscr{L}, \ell)} = \frac{1}{\Eem{\Phi}}.
    \end{equation}
    Substituting~\eqref{eq:norm_constant} into~\eqref{eq:prob_ell_1} gives 
    \begin{equation}\label{eq:prob_ell_2}
         \prob{\ell} =  \frac{\probem{\ell} \Phi(b, \mathscr{L}, \ell)}{\Eem{\Phi(b, \mathscr{L}, \ell)}}.
    \end{equation}
    Taking the expectation of~$\ell$ using the probability mass function in~\eqref{eq:prob_ell_2} gives
    \begin{equation}\label{eq:turn_into_expectation}
        \E{\ell} = \frac{\sum_{\ell= 1}^n \ell\probem{\ell} \Phi(b, \mathscr{L}, \ell)}{\Eem{\Phi(b, \mathscr{L}, \ell)}}.
    \end{equation}
    From the Law of the Unconscious Statistician~\cite{ross2014introduction}, \eqref{eq:turn_into_expectation} becomes
        $\E{\ell} = \frac{\Eem{\ell\Phi(b, \mathscr{L}, \ell)}}{\Eem{\Phi(b, \mathscr{L}, \ell)}}.$ 
    From the definition of the covariance, we find 
    \begin{equation}\label{eq:covariance}
        \E{\ell}  
        =\Eem{\ell} +  \frac{\cov{\ell}{\Phi(b, \mathscr{L}, \ell)}}{\Eem{\Phi(b, \mathscr{L}, \ell)}}.
    \end{equation}
    Since the identity map~$\ell \mapsto \ell$ is strictly increasing
    and~$\Phi(b, \mathscr{L}, \ell)$ is decreasing in~$\ell$ from \cite[Lemma 4]{mckenna2020permute_arxiv}, we have that~$\cov{\ell}{\Phi(b, \mathscr{L}, \ell)}\leq0$ from Chebyshev's sum inequality~\cite{hardy1952inequalities}.     
    Additionally, because~$\Phi(b, \mathscr{L}, \ell)>0$ by definition,     
    ~$\Eem{\Phi(b, \mathscr{L}, \ell) }>0$. As a result, 
        $\E{\ell}\leq \Eem{\ell}.$
    The distribution~$\probem{\ell}$ is equal to~$\Binom(n, q)$ where~$q = \frac{(m-1)e^{-\frac{\epsilon}{2b}}}{1+(m-1)e^{-\frac{\epsilon}{2b}}}$. Then 
        $\Eem{\ell} = \frac{nC}{1+C}, $
    where~$C = (m-1)e^{-\frac{\epsilon}{2b}}$. 
    
    Next we derive the lower bound in the theorem statement. Returning to~\eqref{eq:covariance}, we have 
    $\E{\ell} = \Eem{\ell} +  \frac{\cov{\ell}{\Phi(b, \mathscr{L}, \ell)}}{\Eem{\Phi(b, \mathscr{L}, \ell)}}.$
From Gr\"us's inequality~\cite[Chapter X]{mitrinovic2013classical}, 
we have 
\begin{multline}
    |\cov{\ell}{\Phi(b, \mathscr{L}, \ell)}|\\\leq \frac{1}{4}(\sup_{\ell} \ell - \inf_{\ell} \ell)(\sup_{\ell} \Phi(b, \mathscr{L}, \ell) - \inf_{\ell} \Phi(b, \mathscr{L}, \ell)) \\= \frac{n}{4}(\Phi(b, \mathscr{L}, 0)-\Phi(b, \mathscr{L}, n)).
\end{multline}
Because~$\cov{\ell}{\Phi(b, \mathscr{L}, \ell)}\leq0$, 
\begin{equation}\label{eq:cov}
    \cov{\ell}{\Phi(b, \mathscr{L}, \ell)}\geq -\frac{n}{4}(\Phi(b, \mathscr{L}, 0)-\Phi(b, \mathscr{L}, n)).
\end{equation}
Substituting~$\Eem{\ell}$ and~\eqref{eq:cov} into~\eqref{eq:covariance} gives
    $\E{\ell} \geq \frac{nC}{1+C} -   \frac{n(\Phi(b, \mathscr{L}, 0)-\Phi(b, \mathscr{L}, n))}{4\Eem{\Phi(b, \mathscr{L}, \ell)}}.$
Finally, the random variable~$\ell$ 
is always contained in the interval
$[0, n]$. Thus, from the Hoeffding inequality~\cite{hoeffding1963probability},
    $\prob{|\ell-\E{\ell}|\geq t}\leq 2\exp(-\frac{2t^2}{n^2}). $
\qed

\subsection{Proof of Theorem \ref{thm:mech2}}\label{prf:thm3}
Following a similar approach to Mechanism~\ref{mech:prob1}, we will show Mechanism~\ref{mech:prob3} is word~$\epsilon$-differentially private 
by showing it selects an arbitrary private output
word with the same probability as 
the permute-and-flip mechanism in 
Definition~\ref{def:PF}. From Definition~\ref{def:PF}, the permute-and-flip mechanism selects the private output word~$w'$ with probability
        $\prob{\mathcal{M}_{PF}(w) = w'} = \exp\left(-\frac{\epsilon \ell_w}{2b}\right)\Psi(b,\mathcal{L}(\mathcal{Y}^n), w'),$
    where~$\ell_w$ is the Hamming distance between~$w$ and~$w'$. For Mechanism~\ref{mech:prob3}, the probability of selecting the same private output word~$w'$ is    
    \begin{align}
        \prob{\mathcal{M}_2(w) = w'} &= \prob{\ell_w; w, b}\frac{1}{N_{\mathcal{Y}^n}(\ell)}
        \\ &= \exp\left(-\frac{\epsilon \ell_w}{2b}\right)\Phi(b, \mathscr{L}^{\mathcal{Y}^n}, \ell).\label{eq:rewrite_with_psi}
    \end{align}
    Because~$d(w, w') = d(w, \hat{w}(\ell))$,  where~$\hat{w}(\ell)$ is an arbitrary word with~$\ell$ errors,     
    and $\Psi(b, \mathcal{L}(\mathcal{Y}^n), \hat{w}(\ell)) = \Phi(b, \mathscr{L}^{\mathcal{Y}^n}, \ell)$,     
    we have that~$\Phi(b, \mathscr{L}^{\mathcal{Y}^n}, \ell) = \Psi(b, \mathcal{L}(\mathcal{Y}^n), w')$. Accordingly, we equivalently write~\eqref{eq:rewrite_with_psi}  as 
         $\prob{\mathcal{M}_2(w) = w'} = \exp\left(-\frac{\epsilon \ell_w}{2b}\right)\Psi(b, \mathcal{L}(\mathcal{Y}^n), w').$
    Then~$\prob{\mathcal{M}_{PF}(w) = w'} = \prob{\mathcal{M}_2(w) = w'}$, completing the proof.
    \qed
\bibliographystyle{IEEEtran}
\bibliography{references}

\end{document}